\documentclass[aps, prl, twocolumn]{revtex4}
\pdfoutput=1
\usepackage{amssymb}
\usepackage{amsthm}
\usepackage{amsmath}
\usepackage{epsfig}
\usepackage{float}
\usepackage{hyperref} 
\newtheorem{theorem}{Theorem}

\theoremstyle{definition}
\theoremstyle{definition}
\theoremstyle{definition}\newtheorem{definition}[theorem]{Definition}
\theoremstyle{definition}
\theoremstyle{definition}
\theoremstyle{definition}
\theoremstyle{definition}

\newcommand{\bit}{\begin{itemize}}
\newcommand{\eit}{\end{itemize}\par\noindent}
\newcommand{\ben}{\begin{enumerate}}
\newcommand{\een}{\end{enumerate}\par\noindent}
\newcommand{\beq}{\begin{equation}}
\newcommand{\eeq}{\end{equation}\par\noindent}
\newcommand{\beqa}{\begin{eqnarray*}}
\newcommand{\eeqa}{\end{eqnarray*}\par\noindent}
\newcommand{\beqn}{\begin{eqnarray}}
\newcommand{\eeqn}{\end{eqnarray}\par\noindent}

\newcommand{\sm}[1]{\mbox{\small$#1$}}

\usepackage{graphicx}

\newcommand{\spc}{-0.1in}
\begin{document} 

\title{Time-asymmetry of probabilities versus relativistic causal structure: an arrow of time}
\author{Bob Coecke}
\author{Raymond Lal
}
\affiliation{University of Oxford, Department of Computer Science, Quantum Group
\\ Wolfson Building, Parks Road, Oxford, OX1 3QD, UK.
  }

\begin{abstract}
There is an incompatibility between the symmetries of causal structure in relativity theory and the signaling abilities of probabilistic devices with inputs and outputs: while time-reversal in relativity will not introduce the ability to signal  between spacelike separated regions, this is not the case for probabilistic devices  with space-like separated input-output pairs.  
We explicitly describe a non-signaling device which becomes a perfect signaling device under time-reversal, where  time-reversal can be conceptualized as playing backwards a videotape of an agent manipulating the device. This leads to an arrow of time that is identifiable when studying the correlations of events for spacelike separated regions. Somewhat surprisingly, although time-reversal  of 
Popuscu-R\"orlich boxes also allows agents to signal,
it does not yield  a perfect signaling device.  Finally, we 
realize time-reversal
using post-selection, which could lead experimental 
implementation. 
\end{abstract}

\pacs{...}
\maketitle


The bounded speed of light in relativity theory restricts the ability of agents to signal to other agents.  For this reason, the study of general probabilistic theories \cite{Pitowski, Barrett} is typically restricted to those theories that do not enable signaling, of which classical and quantum probability are of course the main examples.  Due to the symmetries under time-reversal of relativity theory, namely that time-reversal does not introduce the ability to signal  between  space-like separated regions (see also our discussion below), one may expect that under time-reversal the non-signaling constraint would also be preserved by general probabilistic theories.  This turns out not to be the case, not even for classical probability theory.  This has important consequences  if one were to incorporate probabilistic features within causal structure \cite{Causaloid}, and provides important new input into the  debate on the arrow of time \cite{Halliwell}. We discuss this further at the end of this paper.

 
\section{The inconsistency}

Given a spacetime manifold, e.g.~Minkowski spacetime ${\cal M}$,  the  spacelike separation of two regions $A$ and $B$ implies that an agent Alice located in $A$ is  not able to signal to an agent Bob located in $B$, and vice versa. We denote this by $A\not\leq B$ and $B\not\leq A$.  If $A$ and $B$ are time-like separated then signaling is in principle possible, along the direction of time, but of course not backwards. Hence, in the case that $A$ causally precedes $B$ we have $A\leq B$ and $B\not\leq A$.  Now consider the time-reversed spacetime manifold ${\cal M}^{op}$ \footnote{Note that for general relativistic spacetime ${\cal G}$, the spacetime ${\cal G}^{op}$ is also a valid spacetime, since the space of solutions to Einstein's equations is closed under reversal of time-orientation, i.e.~the continuous assignment of `future' and `past' each to one half of the light cone at each point of the manifold; see \cite{Wald}, for further discussion).},  with respect to which  we denote the (in)ability of agents to signal by means of  $\leq^{op}$ and $\not\leq^{op}$.  The symmetry under time-reversal of ${\cal M}$ means that $A\leq B$ if and only if $B\leq^{op} A$ 
(i.e.~when $A$ can signal to $B$ in $\mathcal{M}$), 
while $A\not\leq B$ if and only if $B\not\leq^{op} A$ 
 (i.e.~when $A$ cannot signal to $B$ in $\mathcal{M}$). 
 In particular, for spacelike separated regions (for which $A\not\leq B$ and $B\not\leq A$), time-reversal does not introduce the ability to signal. This argument straightforwardly extends to more general globally hyperbolic spacetimes \footnote{The \em causal structure\em, i.e.~relationships of the kind $A\leq B$ and also $A <B$ between spacetime points, indeed captures much of the essence of relativity: e.g.~Malament \cite{Malament} showed that for globally hyperbolic spacetimes the differentiable structure and the conformal metric can be recovered from the causal structure. 
}.

Now we consider a device with two inputs $a_I$ and $b_I$ and two outputs $a_O$ and $b_O$, Alice and Bob each having access to one input and one output \footnote{Devices of this kind are 
important 
in quantum foundations and quantum information; e.g.~for a Bell-type scenario the inputs constitute the choice of measurements and the outputs the measurement outcomes.}: 
\begin{figure}[H]
\centering
\includegraphics[width=120pt]{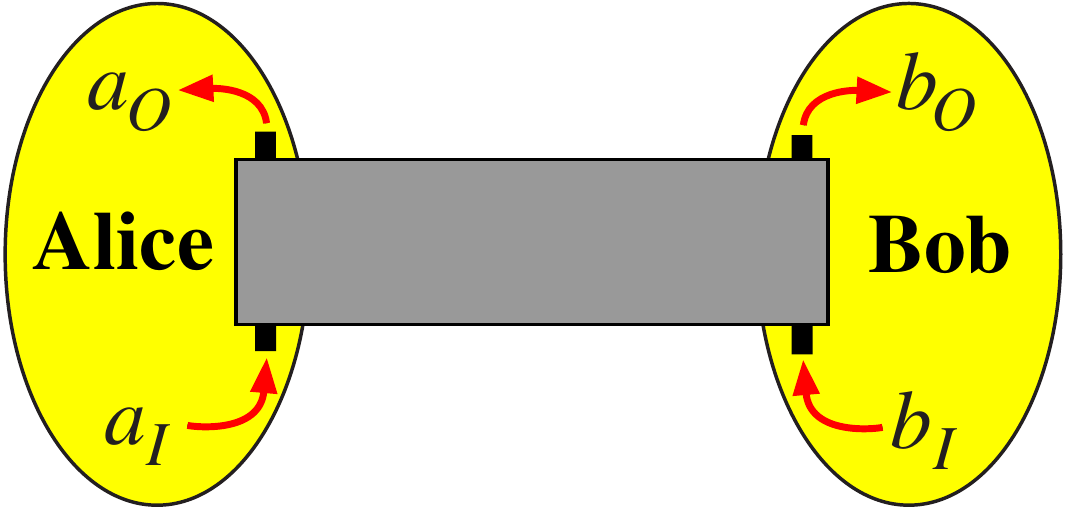}
\caption{Two-party probabilistic input-output box.}
\end{figure}
\vspace{\spc}
We assume that this device is non-signaling,  
and that Alice and Bob remain spatially separated while interacting via the device.  We assume that the device can exhibit any non-signaling correlations: classical (including shared randomness), quantum or super-quantum.
To such a device we associate a time-reversed device by exchanging the roles of the inputs and the outputs; in the next section we discuss how this may be  conceived of operationally. 
We will show that there exist devices for which time-reversal  turns a non-signaling device into a perfect-signaling device. 

So, a \em non-channel in one time direction \em becomes a \em perfect channel in the other direction\em, contra the time-reversal symmetry of relativity discussed above.

One critical issue here is the precise description of time-reversal for such a probabilistic device.  Following standard probability theory one can turn a stochastic matrix with $I$ as givens and $O$ as conclusions into one which has $O$ as givens and $I$ as conclusions via \em Bayesian inversion\em:
\[
P(I|O)={P(O|I) P(I)\over P(O)}.
\]
However this requires  knowing the prior probability distribution $P(I)$, which is supposed to be a free choice by the agents.  We circumvent this issue: our results do not depend whatsoever on the choice of prior.  
More specifically, we will show that only the \em possibility \em of the occurrence of correlations matters, rather than their probability. Then  time-reversal has a unique characterization.

\section{`Detecting' the arrow of time}

We now provide an intuitive operational conception of time-reversal for certain specific situations, which merely consists of rewinding a videotape.  Consider the following multi-agentscenario \footnote{This particular scenario was suggested to us by Jamie Vicary; it improved on one proposed by us which involved illuminated buttons.}.  Each agent has a device with a slot: when one puts a card into the slot, it returns a card.  We consider inserting a card as an input process and retrieving a card as an output process.    The symmetry of the situation now allows for exchanging the roles of the input process and the output process by time-reversal:

\begin{figure}[H]
\centering
\includegraphics[width=240pt]{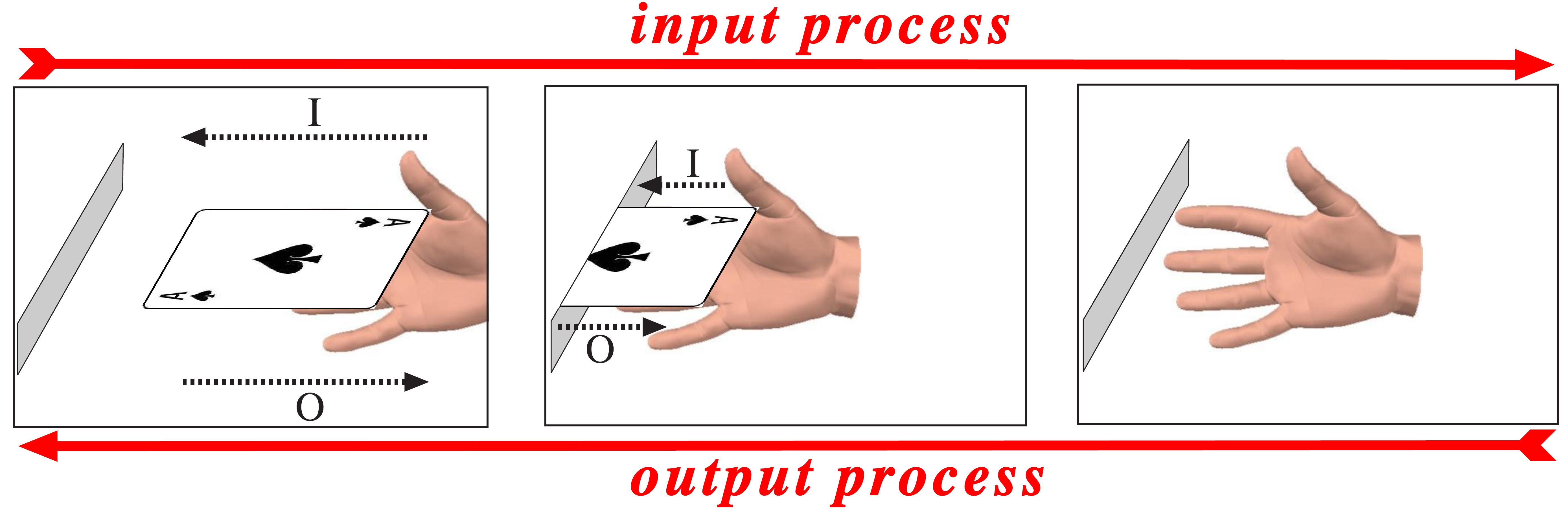}
\caption{Exchange of input and output by time-reversal.}
\label{fig:HandsCard}
\end{figure}
\vspace{\spc}
While the devices are not connected, 
they 
may expose correlations in their behavior which were encoded at the time of their manufacturing. 

Assume that one videotapes the input-output process for a number of rounds that is  sufficient for statistical purposes.  
When one plays the tape backwards, the statistics will now be the one obtained via Bayesian inversion. Analysis of the correlations exposes the fact that these devices enable perfect signaling, despite the fact that the agents may be located outside of each other's light cones. 

We can conclude that merely  by studying correlations one can `detect' the backward direction of time: this is the direction in which \em there exist devices that potentially enable signaling between space-like separated regions\em.  We discuss this issue in more detail at the end of this paper.

\section{Signaling from time-reversal}

Assume that for a two-party I/O-device as described above, $a_I$, $b_I$, $a_O$ and $b_O$ take values in $\{0, 1\}$. We can assign a 4-by-4 \em (probabilistic) correlation matrix \em $g=(g_{a_I, b_I}^{a_O, b_O})$  of which the entries give the probability of obtaining output pair $(a_O, b_O)$ given input pair $(a_I, b_I)$.  

\begin{definition}\label{def:signaling}
A correlation matrix enables  \em  signaling  from Alice to Bob \em iff
\beq\label{eq:probabilistic}
\exists (b_I,b_O):g_{0, b_I}^{0, b_O}+g_{0, b_I}^{1, b_O}\not=g_{1, b_I}^{0, b_O}+g_{1, b_I}^{1, b_O}\,.
\eeq
\end{definition}
The sums reflect that fact that the value of Alice's output is not known to Bob, and hence is traced out.  So by  signaling  we mean that, from his input-output pairs, and after a sufficient number of rounds (i.e.~in the statistical limit), Bob  has obtained information about Alice's sequence of inputs. All  correlation matrices that we will consider here will be invariant under exchange of the roles of Alice and Bob, hence for simplicity we can restrict ourselves to considering only (non-)signaling from Alice to Bob. 

\begin{definition}
A  correlation matrix  $g$ is \em classical \em if there exist 2-by-2 stochastic matrices $\{g_i(a)\}$ and $\{g_i(b)\}$ \footnote{I.e.~$[0,1]$-valued  entries which in each column sum to $1$.} and $\{P_i\}$  for which $P_i\in[0,1]$ and $\sum_i P_i = 1$, such that  $g$ decomposes as follows \footnote{In the quantum foundations literature this is usually referred to as a \em local hidden variable representation\em.}:
\beq
g=\sum_i P_i g_i(a)\times g_i(b)
\eeq
\end{definition}

It is well-known that a classical  correlation matrix does not enable signaling  from Alice to Bob. Now, given $g$ and  prior $P(I)$, we rely on  Bayesian inversion to construct the time-reverse $g_{P(I)}^T$, explicitly, 
\beq
(g_{P(I)}^T)^{a_I, b_I}_{a_O, b_O}= {g_{a_I, b_I}^{a_O, b_O}  \times  (P(I))_{a_I, b_I} \over (P(O))^{a_O, b_O}}
\eeq
where 
\beq\label{eq:posterior}
(P(O))^{a_O, b_O}= \sum_{a_I, b_I} g_{a_I, b_I}^{a_O, b_O} (P(I))_{a_I, b_I}\,.  
\eeq
The variables $a_O$ and $b_O$ are now treated as the inputs and the variables $a_I$ and $b_I$ are now treated as the outputs. By \em perfect signaling \em we mean that Bob receives Alice's input as its output with certainty.  We call a prior \em total \em if it has no zero entries.  

\begin{theorem}\label{thm:classical}
There exist classical correlation matrices for which the time-reverse for any total prior is signaling. More specifically, each such time reverse of 
\beq
\tilde{g}=\left(\begin{array}{cccc}
{1\over 4} & 0 & {1\over 4} & 0\vspace{1mm}\\
{1\over 4} & {1\over 2} & {1\over 4} & {1\over 2}\vspace{1mm}\\
0 & {1\over 4} & 0 & {1\over 4}\vspace{1mm}\\
{1\over 2} & {1\over 4} & {1\over 2} & {1\over 4}
\end{array}\right)
\eeq
enables perfect signaling from Alice to Bob, which is achieved when Bob fixes his input to always be $0$.
\end{theorem}
\begin{proof}
First we observe that $\tilde{g}$ is indeed classical:
\beqa
\tilde{g}&=&{1\over 4}\sm{\left(\begin{array}{cc}
1&  1\\
0&  0
\end{array}\right)}
\!\times\!
\sm{\left(\begin{array}{cc}
1&  0\\
0&  1
\end{array}\right)}
\!+\!
{1\over 4}\sm{\left(\begin{array}{cc}
1& 1\\
0& 0
\end{array}\right)}
\!\times\!
\sm{\left(\begin{array}{cc}
0&  0\\
1&  1
\end{array}\right)}
\\
&&
+
{1\over 4}\sm{\left(\begin{array}{cc}
0&  0\\
1&  1
\end{array}\right)}
\!\times\!
\sm{\left(\begin{array}{cc}
0&  1\\
1&  0
\end{array}\right)}
+
{1\over 4}\sm{\left(\begin{array}{cc}
0&  0\\
1&  1
\end{array}\right)}
\!\times\!
\sm{\left(\begin{array}{cc}
0& 0\\
1& 1
\end{array}\right)}
\eeqa
and hence, in particular, $\tilde{g}$ is non-signaling.  In order to establish that $\tilde{g}_{P(I)}^T$ enables perfect signaling, note that for any $P(I)$ it will always have the form:
\beq
\tilde{g}_{P(I)}^T=\left(\begin{array}{cccc}
a & b & 0 & f\\
0 & c & e & g\\
\!1\mbox{-}a & d & 0 & h\\
0 & \!1\mbox{-}b\mbox{-}c\mbox{-}d & 1\mbox{-}e & \!1\mbox{-}f\mbox{-}g\mbox{-}h\!
\end{array}\right)
\eeq
Assume that Bob fixes his input to be $0$. Then when Alice's input is $0$ the output will be $(0,0)$ with probability $a$ and it will be $(1,0)$ with probability $1-a$, and when Alice's input is $1$ the output will be $(0,1)$ with probability $e$ and it will be $(1,1)$ with probability $1-e$.  Hence, Bob's output always perfectly matches Alice's input.
\end{proof}

\section{The role of Bayesian inversion}

The above discussion explicitly involved Bayesian inversion.  However, the same conclusion can be obtained 
merely by looking at \em possibilities\em, that is, for which pairs of inputs certain outputs are possible. One can represent these possibilities also in a matrix,  with $0$ standing for impossible and $1$ standing for possible.  In this case the time-reverse is nothing but the transpose, e.g.~for the example of Theorem \ref{thm:classical} we have:
\beq\label{mat:classex}
g=\mbox{\small$\left(\begin{array}{cccc}
1 & 0 & 1 & 0\\
1 & 1 & 1 & 1\\ 
0 & 1 & 0 & 1\\
1 & 1 & 1 & 1
\end{array}\right)$}
\ \ \ \ \mbox{}\ \ \ \ \ \
g^T=\mbox{\small$\left(\begin{array}{cccc}
1 & 1 & 0 & 1\\
0 & 1 & 1 & 1\\ 
1 & 1 & 0 & 1\\
0 & 1 & 1 & 1
\end{array}\right)$}.
\eeq
One can show that the reduced data encoded in these matrices is sufficient to draw the same conclusion as stated in Theorem \ref{thm:classical}, but now without any reference to priors.  

\section{The case of PR-boxes} 

One may (wrongly) have assumed that the ability for time-reverses to be signaling may be a result of non-locality.  Since we already established that this phenomenon occurs classically this is not the case. Moreover, for Popuscu-R\"orlich (PR) boxes \cite{PR}, that is, maximally non-local non-signaling correlation matrices, one never achieves perfect signaling under time-reversal.

For a PR-box $pr$ and its time-reverse $pr^T$ we have:
\beq
pr=\mbox{\small$
\left(\begin{array}{cccc}
1 & 1 & 1 & 0\\
0 & 0 & 0 & 1\\
0 & 0 & 0 & 1\\
1 & 1 & 1 & 0
\end{array}\right)$} 
\ \ \ \ \mbox{}\ \ \ \ \ \
pr^T=\mbox{\small$ 
\left(\begin{array}{cccc}
1 & 0 & 0 & 1\\
1 & 0 & 0 & 1\\ 
1 & 0 & 0 & 1\\
0 & 1 & 1 & 0
\end{array}\right)$}.
\eeq
When Bob's output is $0$, then Alice's input always matches Bob's.  Hence in this case he can deduce Alice's input.  However, while we achieve signaling, we don't achieve a perfect channel, since Bob has no control over his output (and neither will Bob instead fixing his input lead to a perfect channel).

\section{Realizing time-reversal}


Above we provided an operational  conception of time-reversal by means of reversing a videotape.  Here we will show how one can effectively \em realize \em time-reversal. Evidently, this will require signaling resources, as the outcome may be a signaling device.  The signaling resource that we will rely on is \em post-selection\em, that is, conditioning on an outcome of a probabilistic process.  

In \cite{LE} it was shown how bipartite states and postselected bipartite measurements can be used to seemingly reverse the flow of quantum data, 
and which was later cast as a diagrammatic formalism that encompassed  many foundational structures of quantum theory \cite{AC,Kindergarten}.  This has been built on by various authors, 
e.g. Svetlichny \cite{Svetlichny}, in proposing post-selected quantum teleportation as a means of simulating closed timelike curves (CTCs) \footnote{Bennett and Schumacher had earlier suggested, in unpublished work, that post-selected teleportation provides a model of time travel.}, which we discuss below.

Now, consider the configuration of Fig. \ref{fig:ABtime}  where the lower triangles represent bipartite states and the upper triangles bipartite effects.
\begin{figure}[H]
\centering
\includegraphics[width=140pt]{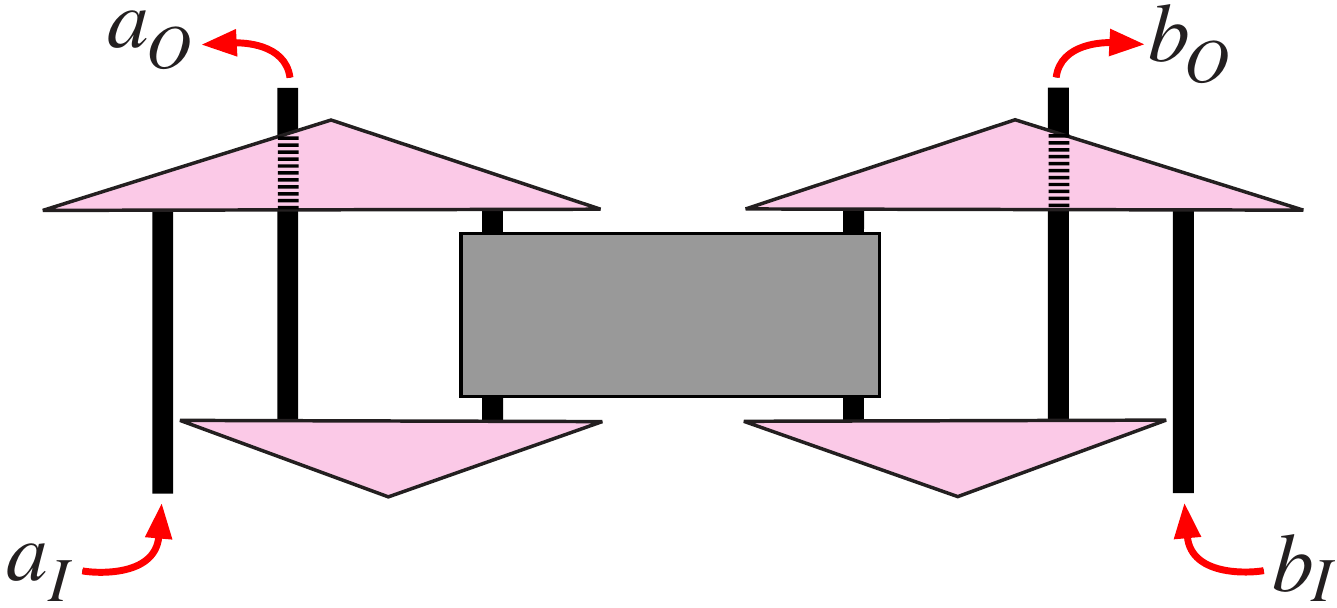}
\caption{Realization of the transpose 
via post-selection.}
\label{fig:ABtime}
\end{figure}
\vspace{\spc}
Following  \cite{Kindergarten}, we can make the apparent `flows' of information explicit by replacing the triangles by wires:
\begin{figure}[H]
\centering
\includegraphics[width=140pt]{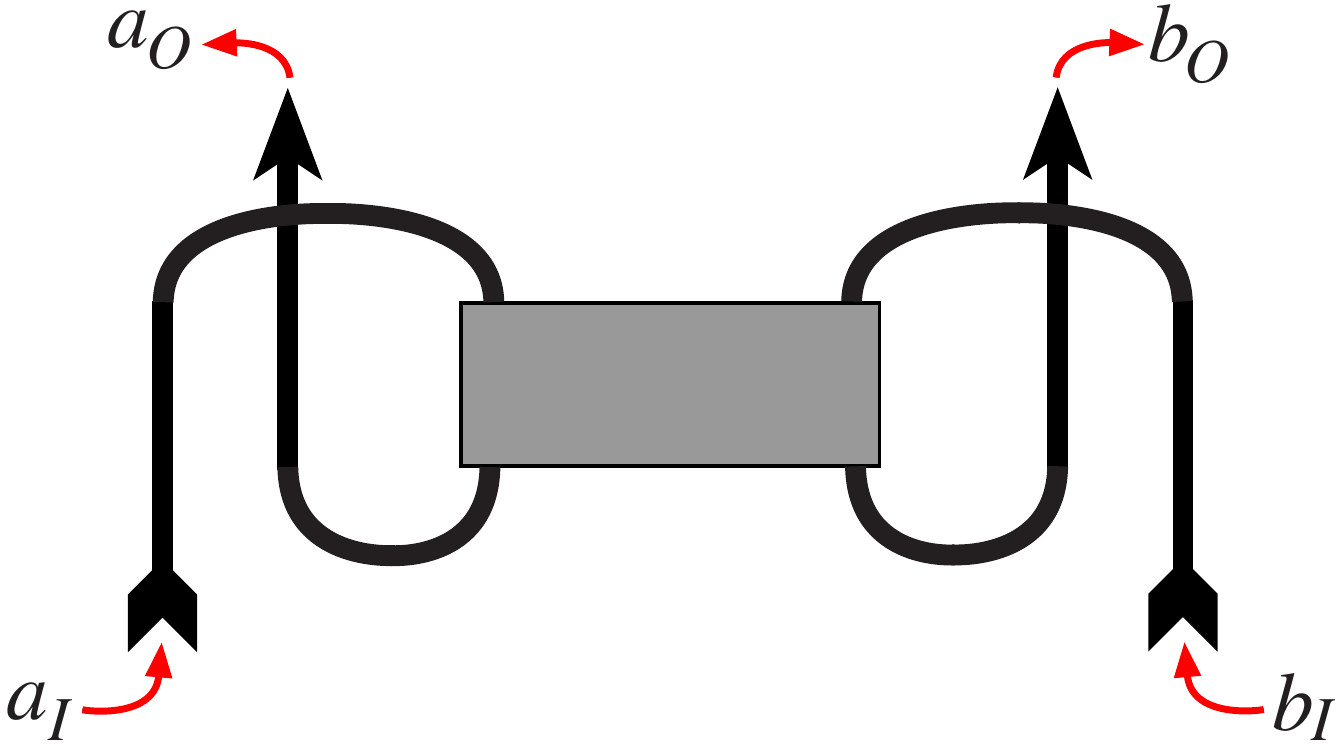}
\caption{Information flow of Fig.~\ref{fig:ABtime}.}
\label{fig:ABflow}
\end{figure}
\vspace{\spc}
Fig.~\ref{fig:ABflow} indeed shows that the inputs $a_I$ and $b_I$ are seemingly `fed' into the outputs of the device. 

This leads us to propose a realization of time-reversal as follows.  Let the states and effects of Fig~\ref{fig:ABtime}  be $\mbox{\small$\left(\begin{array}{cccc}
1 & 0 & 0 & 1
\end{array}\right)$}^T$ and $\mbox{\small$\left(\begin{array}{cccc} 
1 & 0 & 0 & 1
\end{array}\right)$}$ respectively, and consider  the original device (depicted by the grey box)  to be a probabilistic device as described above.  Then it is easily calculated that the  entire post-selected device of Fig~\ref{fig:ABtime} will produce the transpose of the original correlation matrix. Moreover,  Coecke and Spekkens \cite{CS} have shown how the configuration of Fig~\ref{fig:ABtime} can also realize Bayesian inversion: the states and effects that are now used depend on the prior $P(I)$. This produces a realization of time-reversal that is therefore relative to a particular input-output pair of distributions $(P(I),P(O))$.

Note that if the device of Fig.~\ref{fig:ABtime} were to have pairs of qubits as inputs and outputs,
and taking the states and effect respectively to be $|00\rangle+|11\rangle$ and $\langle00|+\langle11|$, then we immediately obtain that transpose of the quantum operation. One could rely on the experimental techniques of \cite{Laflamme, Lloyd}, to effectively realize this in the lab. 
%

Post-selection is hence used both in our proposal for the realization of time-reversal and in the aforementioned proposals for simulation of CTCs \cite{Svetlichny,Lloyd}. Indeed, the operation described above is actually analogous to the simulation of CTCs by post-selected quantum teleportation, which can be realized by
using a state and an effect 
in the following configuration: 
\begin{figure}[H]
\centering
\includegraphics[width=120pt]{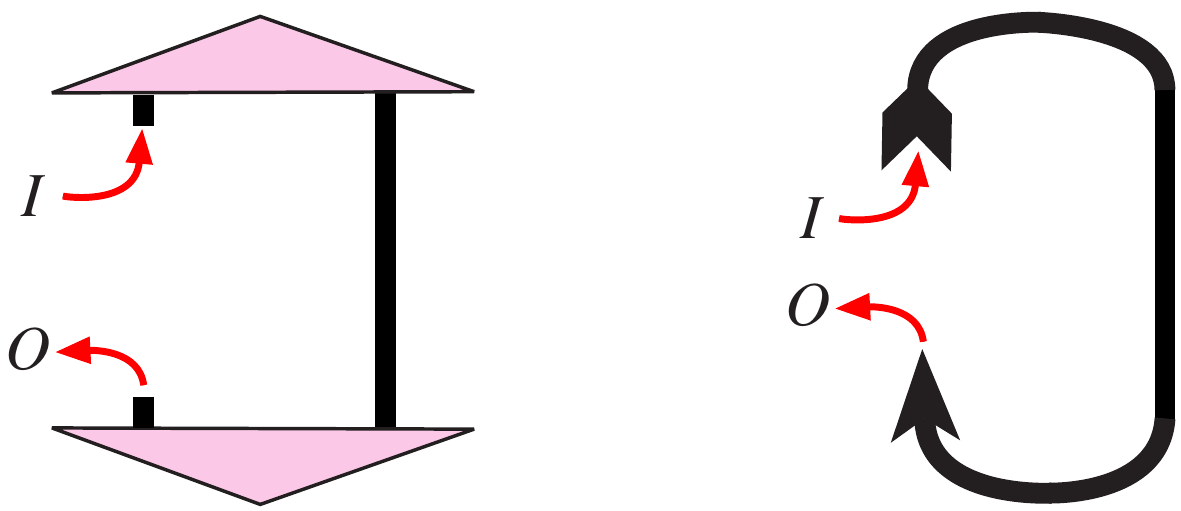}
\caption{Realization of a CTC 
 via post-selection.}
\label{fig:CTC}
\end{figure}
\vspace{\spc}
On the left of Fig.~\ref{fig:CTC}, half of a Bell pair $|00\rangle+|11\rangle$ is subject to part of the entangled effect $\langle00|+\langle11|$: on the right the wires indeed show that this leads to an apparent flow of information backwards in time through a `loop'. 
\section{Discussion}

We showed that when probabilities are involved, the time-reversed picture (cf.~reversing the tape) fundamentally clashes with relativistic light cone structure.  Therefore, it seems that one needs to abandon time-symmetry for any theory that combines both relativistic aspects and 
intrinsic probabilistic aspects.
Dually, an analysis of probabilistic correlations enables one to `detect' an arrow of time. Our point is not to argue for the actual existence of a signaling device that can effectively be realized, but rather that  asserting `too much'  time-related symmetry, as is for example encoded within GR, causes a contradiction with other aspects of physics. 

 In that respect, it should be noted that in the realization using a video tape and reversing it, we never will observe the perfect signaling device being used to signal: we can only deduce from the statistics that it could be potentially used for that purpose, as Bob's backward inputs will typically not always be $0$.  In contrast, in the realization using post-selection we \em will \em observe signaling for the device of Theorem \ref{thm:classical}, but note that in general this will  depend on the chosen $P(I)$: the $P(O)$ we obtain may not lead to backwards signaling. 

 
   For future work, 
it would be worthwhile to investigate how 
our work 
relates to the relationship between information processing and the thermodynamic arrow of time \cite{Maroney}, and also its relation to the signaling properties of PR boxes in the presence of CTCs \cite{Chakra}.  Also, since causal structure forms the basis for approaches to quantum gravity \cite{Sorkin}, these approaches may have to be reconsidered in the light of the results in this paper.

\bibliographystyle{apsrev}

\begin{thebibliography}{17}
\expandafter\ifx\csname natexlab\endcsname\relax\def\natexlab#1{#1}\fi
\expandafter\ifx\csname bibnamefont\endcsname\relax
  \def\bibnamefont#1{#1}\fi
\expandafter\ifx\csname bibfnamefont\endcsname\relax
  \def\bibfnamefont#1{#1}\fi
\expandafter\ifx\csname citenamefont\endcsname\relax
  \def\citenamefont#1{#1}\fi
\expandafter\ifx\csname url\endcsname\relax
  \def\url#1{\texttt{#1}}\fi
\expandafter\ifx\csname urlprefix\endcsname\relax\def\urlprefix{URL }\fi
\providecommand{\bibinfo}[2]{#2}
\providecommand{\eprint}[2][]{\url{#2}}

\bibitem[{\citenamefont{Pitowsky}(1989)}]{Pitowski}
\bibinfo{author}{\bibfnamefont{I.}~\bibnamefont{Pitowsky}},
  \emph{\bibinfo{title}{Quantum Probability, Quantum Logic}}
  (\bibinfo{publisher}{Springer-Verlag, Heidelberg}, \bibinfo{year}{1989}).

\bibitem[{\citenamefont{Barrett}(2007)}]{Barrett}
\bibinfo{author}{\bibfnamefont{J.}~\bibnamefont{Barrett}},
  \bibinfo{journal}{Phys. Rev. A} \textbf{\bibinfo{volume}{75}},
  \bibinfo{pages}{032304} (\bibinfo{year}{2007}), \eprint{quant-ph/0508211}.

\bibitem[{\citenamefont{Hardy}(2007)}]{Causaloid}
\bibinfo{author}{\bibfnamefont{L.}~\bibnamefont{Hardy}},
  \bibinfo{journal}{J.~Phys.~A} \textbf{\bibinfo{volume}{40}},
  \bibinfo{pages}{3081} (\bibinfo{year}{2007}), \eprint{gr-qc/0608043}.

\bibitem[{\citenamefont{Halliwell and Pérez-Mercader}(1994)}]{Halliwell}
\bibinfo{author}{\bibfnamefont{J.~J.} \bibnamefont{Halliwell}}
  \bibnamefont{and}
  \bibinfo{author}{\bibfnamefont{J.}~\bibnamefont{Pérez-Mercader}},
  \emph{\bibinfo{title}{Physical origins of time asymmetry}}
  (\bibinfo{publisher}{Cambridge University Press}, \bibinfo{year}{1994}).

\bibitem[{\citenamefont{{Popescu} and {Rohrlich}}(1994)}]{PR}
\bibinfo{author}{\bibfnamefont{S.}~\bibnamefont{{Popescu}}} \bibnamefont{and}
  \bibinfo{author}{\bibfnamefont{D.}~\bibnamefont{{Rohrlich}}},
  \bibinfo{journal}{Found.~Phys.} \textbf{\bibinfo{volume}{24}},
  \bibinfo{pages}{379} (\bibinfo{year}{1994}).

\bibitem[{\citenamefont{Coecke}(2003)}]{LE}
\bibinfo{author}{\bibfnamefont{B.}~\bibnamefont{Coecke}}
  (\bibinfo{year}{2003}), \eprint{quant-ph/0402014}.

\bibitem[{\citenamefont{Abramsky and Coecke}(2004)}]{AC}
\bibinfo{author}{\bibfnamefont{S.}~\bibnamefont{Abramsky}} \bibnamefont{and}
  \bibinfo{author}{\bibfnamefont{B.}~\bibnamefont{Coecke}}, in
  \emph{\bibinfo{booktitle}{Proceedings of the 19th Annual IEEE Symposium on
  Logic in Computer Science}} (\bibinfo{publisher}{IEEE Computer Society},
  \bibinfo{address}{Washington, DC, USA}, \bibinfo{year}{2004}), pp.
  \bibinfo{pages}{415--425}, ISBN \bibinfo{isbn}{0-7695-2192-4},
  \urlprefix\url{http://portal.acm.org/citation.cfm?id=1018438.1021878}.

\bibitem[{\citenamefont{Coecke}(2006)}]{Kindergarten}
\bibinfo{author}{\bibfnamefont{B.}~\bibnamefont{Coecke}}, \bibinfo{journal}{AIP
  Conference Proceedings} \textbf{\bibinfo{volume}{810}}, \bibinfo{pages}{81}
  (\bibinfo{year}{2006}),
  \urlprefix\url{http://link.aip.org/link/?APC/810/81/1}.

\bibitem[{\citenamefont{Svetlichny}(2009)}]{Svetlichny}
\bibinfo{author}{\bibfnamefont{G.}~\bibnamefont{Svetlichny}}
  (\bibinfo{year}{2009}), \eprint{0902.4898}.

\bibitem[{\citenamefont{Coecke and Spekkens}(2011)}]{CS}
\bibinfo{author}{\bibfnamefont{B.}~\bibnamefont{Coecke}} \bibnamefont{and}
  \bibinfo{author}{\bibfnamefont{R.~W.} \bibnamefont{Spekkens}},
  \bibinfo{journal}{Synthese} pp. \bibinfo{pages}{1--46}
  (\bibinfo{year}{2011}), ISSN \bibinfo{issn}{0039-7857},
  \bibinfo{note}{10.1007/s11229-011-9917-5},
  \urlprefix\url{http://dx.doi.org/10.1007/s11229-011-9917-5}.

\bibitem[{\citenamefont{Laforest et~al.}(2006)\citenamefont{Laforest, Baugh,
  and Laflamme}}]{Laflamme}
\bibinfo{author}{\bibfnamefont{M.}~\bibnamefont{Laforest}},
  \bibinfo{author}{\bibfnamefont{J.}~\bibnamefont{Baugh}}, \bibnamefont{and}
  \bibinfo{author}{\bibfnamefont{R.}~\bibnamefont{Laflamme}},
  \bibinfo{journal}{Phys. Rev. A} \textbf{\bibinfo{volume}{73}},
  \bibinfo{pages}{032323} (\bibinfo{year}{2006}).

\bibitem[{\citenamefont{Lloyd et~al.}(2011)\citenamefont{Lloyd, Maccone,
  Garcia-Patron, Giovannetti, Shikano, Pirandola, Rozema, Darabi, Soudagar,
  Shalm et~al.}}]{Lloyd}
\bibinfo{author}{\bibfnamefont{S.}~\bibnamefont{Lloyd}},
  \bibinfo{author}{\bibfnamefont{L.}~\bibnamefont{Maccone}},
  \bibinfo{author}{\bibfnamefont{R.}~\bibnamefont{Garcia-Patron}},
  \bibinfo{author}{\bibfnamefont{V.}~\bibnamefont{Giovannetti}},
  \bibinfo{author}{\bibfnamefont{Y.}~\bibnamefont{Shikano}},
  \bibinfo{author}{\bibfnamefont{S.}~\bibnamefont{Pirandola}},
  \bibinfo{author}{\bibfnamefont{L.~A.} \bibnamefont{Rozema}},
  \bibinfo{author}{\bibfnamefont{A.}~\bibnamefont{Darabi}},
  \bibinfo{author}{\bibfnamefont{Y.}~\bibnamefont{Soudagar}},
  \bibinfo{author}{\bibfnamefont{L.~K.} \bibnamefont{Shalm}},
  \bibnamefont{et~al.}, \bibinfo{journal}{Phys. Rev. Lett.}
  \textbf{\bibinfo{volume}{106}}, \bibinfo{pages}{040403}
  (\bibinfo{year}{2011}).

\bibitem[{\citenamefont{Maroney}(2010)}]{Maroney}
\bibinfo{author}{\bibfnamefont{O.}~\bibnamefont{Maroney}},
  \bibinfo{journal}{Foundations of Physics} \textbf{\bibinfo{volume}{40}},
  \bibinfo{pages}{205} (\bibinfo{year}{2010}), ISSN \bibinfo{issn}{0015-9018},
  \urlprefix\url{http://dx.doi.org/10.1007/s10701-009-9386-6}.

\bibitem[{\citenamefont{{Chakrabarty} et~al.}(2011)\citenamefont{{Chakrabarty},
  {Pati}, and {Agrawal}}}]{Chakra}
\bibinfo{author}{\bibfnamefont{I.}~\bibnamefont{{Chakrabarty}}},
  \bibinfo{author}{\bibfnamefont{A.~K.} \bibnamefont{{Pati}}},
  \bibnamefont{and}
  \bibinfo{author}{\bibfnamefont{P.}~\bibnamefont{{Agrawal}}},
  \bibinfo{journal}{ArXiv e-prints}  (\bibinfo{year}{2011}),
  \eprint{1107.2908}.

\bibitem[{\citenamefont{Sorkin}(2003)}]{Sorkin}
\bibinfo{author}{\bibfnamefont{R.~D.} \bibnamefont{Sorkin}},
  \emph{\bibinfo{title}{Causal sets: Discrete gravity (notes for the valdivia
  summer school)}} (\bibinfo{year}{2003}), \eprint{gr-qc/0309009}.

\bibitem[{\citenamefont{Wald}(1980)}]{Wald}
\bibinfo{author}{\bibfnamefont{R.~M.} \bibnamefont{Wald}},
  \bibinfo{journal}{Phys. Rev. D} \textbf{\bibinfo{volume}{21}},
  \bibinfo{pages}{2742} (\bibinfo{year}{1980}).

\bibitem[{\citenamefont{Malament}(1977)}]{Malament}
\bibinfo{author}{\bibfnamefont{D.~B.} \bibnamefont{Malament}},
  \bibinfo{journal}{J.~Math.~Phys.} \textbf{\bibinfo{volume}{18}},
  \bibinfo{pages}{1399} (\bibinfo{year}{1977}).

\end{thebibliography}

\end{document}